\pdfoutput=1 

\newif\ifconference
\conferencetrue

\ifconference
\documentclass[twoside,leqno,11pt]{article}
\usepackage{ltexpprt}
\usepackage{hyperref}
\else
\documentclass[11pt]{article}
\usepackage[hyperindex,breaklinks,bookmarks]{hyperref}
\usepackage{amsthm}
\fi
\usepackage[utf8]{inputenc}
\usepackage[english]{babel}
\usepackage{amsmath,amsfonts}
\usepackage{amssymb}
\usepackage[showonlyrefs]{mathtools}
\usepackage{moresize}
\usepackage[a4paper, portrait, left=1in, right=1in, top=1in, bottom=1in]{geometry}
\usepackage[numbers,sectionbib,longnamesfirst]{natbib}
\hypersetup{bookmarksdepth=2}
\usepackage[shortlabels]{enumitem}
\usepackage[toc]{appendix}
\usepackage{microtype}
\usepackage[ruled,linesnumbered]{algorithm2e}
\usepackage{cleveref}
\usepackage{xcolor}

\newcommand{\comments}{1} 
\ifnum\comments=1
\else
\newcommand{\jakub}[1]{}
\newcommand{\talya}[1]{}
\newcommand{\shyam}[1]{}
\fi

\SetKwFor{RepTimes}{repeat}{times}{end}
\SetKwFor{Repeat}{repeat}{}{end}
\SetKwIF{WithProb}{somethin}{}{with probability}{}{}{end}{otherwise}

\ifconference \else
\newtheorem{theorem}{Theorem}
\newtheorem{corollary}[theorem]{Corollary}

\newtheorem{lemma}[theorem]{Lemma}

\theoremstyle{remark}

\fi

\newcommand{\eps}{\varepsilon}
\renewcommand{\epsilon}{\varepsilon}
\newcommand{\BE}{\mathbb{E}}
\newcommand{\BP}{\mathbb{P}}

\date{}
\ifconference
\title{\Large Sampling an Edge in Sublinear Time Exactly and Optimally
}
\else
\title{Sampling an Edge in Sublinear Time Exactly and Optimally
}
\fi

\ifconference
\author{Talya Eden\thanks{Bar Ilan University, \texttt{talyaa01@gmail.com}, supported by the NSF TRIPODS program, award CCF-1740751. This work was partially done while affiliated with MIT and Boston University.} \and Shyam Narayanan\thanks{MIT, \texttt{shyamsn@mit.edu}, supported by the NSF GRFP Fellowship and the NSF TRIPODS Program (award DMS-2022448).} \and Jakub T\v{e}tek\thanks{BARC, Univ. of Copenhagen, \texttt{j.tetek@gmail.com}, supported by the VILLUM Foundation grant 16582. This work was partially done while visiting MIT.}}
\else
\author{
    Talya Eden\thanks{Talya Eden is supported by the NSF TRIPODS program, award CCF-1740751. This work was partially done while affiliated with MIT and Boston University.}\\ \texttt{talyaa01@gmail.com} \\ Bar Ilan University
    \and
    Shyam Narayanan\thanks{Shyam Narayanan is supported by the NSF GRFP Fellowship and the NSF TRIPODS Program (award DMS-2022448).}\\ \texttt{shyamsn@mit.edu} \\ 
    MIT
    \and
	Jakub Tětek\thanks{Jakub T\v{e}tek is supported by the VILLUM Foundation grant 16582. This work was partially done while visiting MIT.}\\
	\texttt{j.tetek@gmail.com}\\
	BARC, Univ. of Copenhagen
}
\fi

\begin{document}
\maketitle
\ifconference
\fancyfoot[R]{\scriptsize{Copyright \textcopyright\ 2023 by SIAM\\
Unauthorized reproduction of this article is prohibited}}
\else
\thispagestyle{empty}
\fi

\begin{abstract}
\ifconference \small\baselineskip=9pt \fi
Sampling edges from a graph in sublinear time is a fundamental problem and a powerful subroutine for designing sublinear-time algorithms. Suppose we have access to the vertices of the graph and know a constant-factor approximation to the number of edges. An algorithm for pointwise $\varepsilon$-approximate edge sampling with complexity $O(n/\sqrt{\varepsilon m})$ has been given by Eden and Rosenbaum [SOSA 2018]. This has been later improved by T\v{e}tek and Thorup [STOC 2022] to $O(n \log(\varepsilon^{-1})/\sqrt{m})$. At the same time, $\Omega(n/\sqrt{m})$ time is necessary. We close the problem, by giving an algorithm with complexity $O(n/\sqrt{m})$ for the task of sampling an edge exactly uniformly.

%Our algorithm is based on a new technique that we call \emph{Bernoulli trial simulation}. We believe this technique could also be useful for other problems. 
%
%Given access to trials of the form $Bern(p)$, this technique allows us to simulate a Bernoulli trial $Bern(f(p) \pm \varepsilon)$ (without knowing $p$), in time complexity $O(\log \varepsilon^{-1})$ for some functions $f$. We specifically use this for $f(p) = 1/(2p)$ for $p \geq 2/3$. Therefore, we can perform rejection sampling, without the algorithm having to know the desired rejection probability. We conjecture Bernoulli trial simulation for $f(p) = 1/(2p)$ can be done exactly in expected $O(1)$ samples. This would lead to an exact algorithm for sampling an edge with complexity $O(n/\sqrt{m})$, completely resolving the problem of sampling an edge, again assuming rough knowledge of $m$. We consider the problem of removing this assumption to be an interesting open problem.
\end{abstract}

\newpage
\pagenumbering{arabic}

\section{Introduction}
Suppose we have a graph too big to even read the whole input. We then need an algorithm running in time sublinear in the input size. Such algorithms %are often called \emph{sublinear-time graph algorithms} and
have recently received a lot of attention. In the sublinear-time settings, one usually has direct access to the vertices of the input graph, but not to the edges. Because of this, one tool commonly used for designing sublinear-time graph algorithms is an algorithm for sampling edges. This allows us to design an algorithm that uses random edge queries, as we can simulate these queries by the edge sampling algorithm.

%Formally, we want to sample an edge $\varepsilon$-pointwise-close to uniform. That is, assuming there are $m$ edges, we want to sample each edge with probability in $[\frac{1-\varepsilon}{m},\frac{1+\varepsilon}{m}]$. \footnote{In general, we say distribution $D_1$ is $\varepsilon$-pointwise-close to distribution $D_2$ if, for any measurable $S$, it holds $1-\varepsilon \leq D_1(S)/D_2(S) \leq 1+\varepsilon$} We discuss the importance of using pointwise closeness below.
%
\paragraph{The task and query access.} Our goal is to sample an edge uniformly, i.e., to return an edge so that each edge is returned with exactly equal probability.
%It was previously only known how to do this approximately. Our algorithm is not only exact but also more efficient. Specifically, the complexity of our algorithm is the same as that of the previous state-of-the-art for constant value of $\epsilon$ where $\epsilon$ is the error parameter in the approximate version of the problem (defined below).
%
We assume that the algorithm may $(i)$ ask for the $i$-th vertex of the input graph, $(ii)$ ask for the degree of a given vertex, and $(iii)$ ask for the $j$-th neighbor of a given vertex. We assume the algorithm has (approximate) knowledge of the number of edges $m$. This assumption of knowing $m$ was not made in the previous work, and we think getting rid of this assumption is a very interesting open problem. This assumption is, however, not a barrier to using our algorithm as a subroutine for implementing random edge queries, as we discuss below.

\paragraph{Ours and previous results.} In past work  only algorithms for sampling $\epsilon$-\emph{approximately} uniformly in the pointwise distance (or equivalently approximately in the $\ell_\infty$ metric) were given, where the state-of-the-art complexity for that problem was $O(n \log (\varepsilon^{-1})/\sqrt{m})$ given by T\v{e}tek and Thorup~\cite{tetek2022edge}.
 Our algorithm is not only exact, but also more efficient. Specifically, the expected complexity of our algorithm is $O(n/\sqrt{m})$.
% Previously this was only known for constant values of $\epsilon$, as the state-of-the-art complexity of the problem was $O(n \log (\epsilon^{-1})/\sqrt{m})$.
%which is the same the same as that of the previous state-of-the-art complexity of $O(n \log (\epsilon^{-1})/\sqrt{m})$ for a constant value of $\epsilon$. 
This is known to be the best possible. If we have a graph with $n - \Theta(\sqrt{m})$ isolated vertices and a clique over $\Theta(\sqrt{m})$ vertices with $m$ edges, we need to sample $\Omega(n/\sqrt{m})$ vertices before we expect  to see a single edge, giving us a simple matching lower bound.

\paragraph{Using our algorithm as a subroutine and the necessity of knowing $m$.}
As we said above, our algorithm needs to have a constant-factor approximation of the number of edges. This was not necessary in  previous works. The reason is that the previous state-of-the-art has complexity in which one can also afford to independently estimate the number of edges; the possibility of the estimate being incorrect is then added to the bias of the edge sampling algorithm. However, as our algorithm is more efficient, we are not able to estimate $m$ in that complexity, and since we want to sample exactly, we cannot accept that the estimate could be wrong.

Suppose we have an algorithm $A$ that performs random edge queries. We may then use our algorithm in a black box manner to implement these queries (unlike, for example, the algorithm for sampling multiple edges from \cite{eden_multiple} which has polynomial dependency on $\varepsilon$). Specifically, if $A$ uses $q$ random edge queries, then we may set $\varepsilon = 1/(10 q)$ and it will only decrease the success probability of $A$ by at most $1/10$. \footnote{This holds because the total variation distance from uniform of each query is at most $1/(10 q)$, so the total variation distance from uniform of the sequence of $q$ queries is at most $1/10$, meaning that the output from the algorithm has total variation distance at most $1/10$ from the distribution the output would have if the queries were answered exactly.}

If the goal is to get an algorithm with a constant success probability (which can then be amplified) that uses our edge sampling algorithm as a subroutine, then we may remove the need for having an a priori constant-factor approximation $\tilde{m}$ of $m$ by computing it using the algorithm from \cite{feige2006sums}, only adding a constant to the failure probability. The algorithm from \cite{feige2006sums} has expected complexity $O(n/\sqrt{m})$. The complexity of our algorithm will also still be as desired: it follows from our analysis that the complexity is $O(\frac{n \sqrt{\tilde{m}}}{m})$. It holds by the Jensen inequality that $\BE[\frac{n \sqrt{\tilde{m}}}{m}] \leq \frac{n \sqrt{\BE[\tilde{m}}]}{m} + \log \varepsilon^{-1} = O(n/\sqrt{m})$ since it holds $\BE[\tilde{m}] = O(m)$.

To summarize, we may remove the assumption of a priori knowledge of $m$ when sampling multiple edges at the cost of adding a constant failure probability. This means that we may use our algorithm as a subroutine in an algorithm with constant probability of error, even without knowing $m$ a priori.

\subsection{Technical overview}\label{sec:overview}
% The starting point of our algorithm, is the algorithm by \citet{eden2018sampling}, which has a multiplicative dependence in $1/\sqrt{\eps}$. We then observe how this dependency could be made additive, but with a factor of $1/\eps^2$. We  finally we get to the main contribution of the paper which is proving that the polynomial dependency could be replaced by a logarithmic one, using Bernoulli trial simulation.  

% We first describe the algorithm by \citet{eden2018sampling} which is the starting point of our algorithm. We then explain our idea for replacing 

The starting point of our algorithm is the algorithm by \citet{eden2018sampling}, which we now shortly recall.% We then describe how Bernoulli trial simulation can be used to get a more efficient algorithm. We describe the techniques behind the Bernoulli trial simulation -- our main technical contribution -- at the end of this section.

\subsubsection{The algorithm by \citet{eden2018sampling}.}
Consider each undirected edge as two directed edges, and let $\theta$ be a degree threshold. We refer to vertices with degree at most $\theta$ as \emph{light} vertices, and to all other vertices as \emph{heavy}. We refer to edges originating in light vertices as \emph{light edges}, and to all other edges as \emph{heavy edges}. Using rejection sampling, light edges can be sampled with probability \emph{exactly} $\frac{1}{n\theta}$: by sampling a uniform vertex $v$, then sampling one of its incident edges u.a.r., and then returning that edge with probability $\frac{d(v)}{\theta}$. Sampling heavy vertices is done by first sampling a light edge $uv$ as described above, and if the second endpoint $v$ of the sampled light edge is heavy, sampling one of its incident edges. This procedure results  in  every heavy edge $vw$ being  sampled with probability $\frac{d_{\ell}(v)}{n\theta}\cdot \frac{1}{d(v)}$, where $d_{\ell}(v)$ is the number of light neighbors of $v$. In \citet{eden2018sampling}, $\theta$ is set to $\sqrt{2m/\eps}$ 
which implies that for every heavy vertex $v$, $d_{\ell}(v)\in [(1-\eps)d(v), d(v)]$.\footnote{Since, denoting by $\cal{H}$ the set of heavy vertices, and by $d_{h}(v)$ the number of heavy neighbors of vertex $v$, we have the following. For every $v\in \cal{H}$, $d_{h}(v)\leq |H|\leq \frac{2m}{\theta}=\sqrt{2\eps m}=\eps \theta\leq \eps d(v)$. Therefore, $d_{\ell}(v)>(1-\eps)d(v)$.}
%In our approach, we set in the algorithm from \cite{eden2018sampling} $epsilon = \Theta(1)$, and
Hence, each (heavy) edge is sampled with probability in $[\frac{(1-\epsilon)}{n \theta}, \frac{1}{n \theta}]$. The total probability of sampling some edge (with the algorithm failing otherwise) is thus at least $\frac{(1-\epsilon)m}{n \theta} \approx \frac{\sqrt{\epsilon m}}{n}$. Therefore, the number of attempts needed before we expect to sample an edge is $O(\frac{n}{\sqrt{\epsilon m}})$, implying  a multiplicative dependence on $\eps$.

\subsubsection{Improving the dependency on $\eps$.}
In order to avoid the multiplicative dependency in $\eps$, we instead  set  the threshold $\theta$ to $\sqrt{c m}$ for some constant $c$.
Considering the same sampling procedures as before,  light edges can still be sampled with probability exactly $\frac{1}{n\theta}$. For heavy edges, however, 
% there is constant bias in their sampling probability, since 
the values $d_{\ell}(v)/d(v)$ can vary up to a constant factor between the different 
 heavy vertices,  %Hence, there is 
 leading to a large bias towards heavy edges originating in vertices $v$ with higher values of $\frac{d_{\ell}(v)}{d(v)}$. 
% We can set the value of $\theta=\Theta(\sqrt m)$ to get that for each heavy vertex, $d_{\ell}(v)> \frac{1}{3}d(v)$. 
If  for each vertex $v$, we knew the value of $d_{\ell}(v)$, we could use rejection sampling with probability $q$ that is inversely proportional to $p=\frac{d_\ell(v)}{d(v)}$, e.g., $q = \frac{d(v)}{2d_{\ell(v)}}=\frac{1}{2p}$ (we may assume that, say, $p\geq 2/3$ and thus $q < 1$, by making $c$ large enough). This would result in each heavy edge being sampled with exactly equal probability $\frac{d_{\ell(v)}}{n\theta}\cdot \frac{d(v)}{2d_{\ell}(v)}=\frac{1}{2n\theta}$.

While we do not know the exact value of $d_{\ell}(v)$, we can approximate it up to a $(1\pm\Theta(\eps))$-multiplicative factor using $O(1/\eps^2)$ neighbor queries. This results in $(1\pm \Theta(\epsilon))$-approximation of $q$ and thus leads to a distribution that is $\eps$-close to uniform.  Note that  we only need to approximate $q$ when the  algorithm samples a heavy edge. Moreover, when we do that, we return the edge with constant probability. Thus, in expectation, we only need to approximate $q$  a constant number of times. This means that the total expected time complexity is $O(n/\sqrt{m} + 1/\varepsilon^2)$.

%We now reach the main contribution of the paper, which is achieving an additive logarithmic dependency. 
To remove the dependence on $1/\eps^2$, our main observation is that we do not actually need to (approximately) learn the value of $p$, in order to reject with probability proportional to $q=1/(2p)$. Rather, we can ``simulate'' a Bernoulli trial with probability exactly $\frac{1}{2p}$ by using the results of only $O(1)$ many $Bern(p)$ trials in expectation (though possibly more in the worst case), using the Bernoulli Factory technique of Nacu and Peres~\cite{nacu2005simulation}.

% We use Bernoulli trial simulation as follows. If we sample an edge $uv$ for $v$ being a heavy vertex, the bias is $  $

%\subsubsection{An $O(n/\sqrt{m} + \log\varepsilon^{-1})$ algorithm}

When we sample a uniform neighbor of a heavy vertex $v$, we see a light neighbor of $v$ with probability exactly $p=\frac{d_{\ell}(v)}{d(v)}$, where, as discussed above, we can set $\theta$ so that $p>2/3$.
Therefore, we have access to a Bernoulli trial that succeeds with probability $Bern(p)$ for $p>2/3$. 
As previously explained, in order to achieve  uniformity, we need to perform  rejection sampling (corresponding to a Bernoulli trial) that succeeds with probability $=\frac{1}{2p}$. We can simulate $Bern(1/(2p))$ by relying on the results of an expected $O(1)$ independent copies of $Bern(p)$. 
Namely, we perform an expected $O(1)$ neighbor queries where each results in a light neighbor with probability $p$, giving us an independent copy of a random variable distributed as $Bern(p)$. If we knew that $p$ is bounded away from $1$, we could directly use the result of \citet{nacu2005simulation}. We can ensure that this is the case by first rejecting with probability, say, $1/2$. The probability of getting a light neighbor and not rejecting is then $p/2 \in [1/3,2/3]$. We then use the result of \cite{nacu2005simulation} with the function $1/(4x)$, thus simulating $Bern(1/(4p/2))=Bern(1/(2p))$.

\subsection{Related work}

Using uniform edge samples as a basic query in the sublinear time setting was first suggested by Aliakbarpour et al.~\cite{aliakbarpour} in the context of estimating the number of $s$-stars in a graph, where they showed that this access allows to circumvent lower bounds that hold in the standard adjacency list access. 
It was later used for the more general tasks of estimating and uniformly sampling arbitrary subgraphs in sublinear time~\cite{AKK19, Peng20, BER21,tetek2022edge}.

As mentioned in the introduction, sampling edges from a distribution that is pointwise close to uniform in sublinear time was first suggested by Eden and Rosenbaum~\cite{eden2018sampling} who gave an algorithm with complexity $O\left(n/\sqrt{\varepsilon m}\right)$. This was later improved by T\v{e}tek and Thorup~\cite{tetek2022edge} to an algorithm with complexity $O(n\log(\varepsilon^{-1})/\sqrt m)$.
They also considered two additional access models (full neighborhood access and hash-ordered access) and gave new lower and upper bounds for these settings. In~\cite{eden_multiple}, Eden, Mossel, and Rubinfeld gave an upper bound for the problem of sampling $k$ edges from a pointwise close to uniform distribution. The complexity of their algorithm is $O\left(\sqrt{k}\cdot \frac{n}{\sqrt {m}}\cdot \frac{\log^2n}{\varepsilon^{2.5}}+k\right)$. This was later shown to be essentially optimal (i.e., up to the dependencies on $\varepsilon$ and $\log n$) by~\cite{tetek2022edge}.
In~\cite{eden2019arboricity}, Eden, Ron and Rosenbaum gave an $O\left(\frac{n\alpha}{m}\cdot\frac{\log ^3 n}{\varepsilon}\right)$ algorithm for sampling edges in graphs with arboricity at most $\alpha$. They also showed their algorithm is optimal up to the $\text{poly}(\log n,\varepsilon^{-1})$ dependencies.

The task of sampling close to uniform edges was also recently considered in the setting where the access to the graph is given via  Bipartite Independent Set (BIS) queries~\cite{lapinskas2019approximately,bhattacharya2022faster, addanki2022non}.

\section{Preliminaries}

\paragraph{The query model:} Since we do not have time to read the whole input (and thus to change its representation), it matters how exactly we are able to query it. Throughout this paper, we assume that the graphs' vertices are labeled arbitrarily in $[n]$, the edges of every vertex $v$ are labeled arbitrarily by $[d(v)]$, and that the algorithm knows $n$.
We then assume the following standard set of queries:
\begin{itemize}
    \item \textbf{Uniform vertex queries:} given $i \in [n]$, return the $i$-th vertex
    \item \textbf{Degree queries:} given a vertex $v$, return its degree $d(v)$
    \item \textbf{Neighbor queries:} given a vertex $v$ and $j \in [d(v)]$, return the $j$-th neighbor of $v$
\end{itemize}
This setting has been previously called the adjacency list or indexed neighbor access model and is among the most-studied settings for sublinear-time algorithms. 

A model with an additional query
\begin{itemize}
    \item \textbf{Uniform edge queries:} given $i \in [m]$, return vertices $u,v$ such that $uv$ is the $i$-th edge of the graph
\end{itemize}
has also been considered. Our algorithm can be thought of as a reduction between the two models. One can show (by randomly permuting the edges) that up to a logarithmic factor this setting is equivalent to just assuming random edge queries (with replacement). Our algorithm then allows us to simulate  random edge queries in the indexed neighborhood access model.

\section{Sampling an edge}

In this section, we give our algorithm for sampling an edge. We start by stating a result of \citet{nacu2005simulation} about ``Bernoulli factories''. We recall that a function $f$ from a closed interval $I$ to $\mathbb{R}$ is said to be \emph{real analytic} if for any point $x_0$ in the interior of $I$, $f(x) = \sum_{n=0}^\infty \frac{f^{(n)}(x_0)}{n!} \cdot (x-x_0)^n$ where $f^{(n)}(x_0)$ represents the $n$th derivative of $f$ at $x_0$. Equivalently, $f$ matches its Taylor series about $x_0$ for all of $I$.

\begin{theorem} \label{thm:coin_simulator}
Let $I \subset (0,1)$ be a closed interval and $f: I \rightarrow (0,1)$ be a real analytic function. Let $p$ be a number in $I$. Then, there exists an algorithm independent of $p$ that performs in expectation $O(1)$ independent trials from $Bern(p)$ and returns one Bernoulli trial with distribution $Bern(f(p))$. In addition, the probability of using more than $k$ independent trials is at most $C \rho^k$, for some constants $C \ge 1, 0 < \rho < 1$ that only depend on $I$.
\end{theorem}

Theorem \ref{thm:coin_simulator} gives us the following corollary.
\begin{corollary} \label{cor:simulate}
There is an algorithm that for any $p \in [2/3,1]$ performs in expectation $O(1)$ trials from $Bern(p)$ and returns one Bernoulli trial distributed as $Bern(1/(2p))$. (Note that the algorithm may also use its own randomness, independent of the $Bern(p)$'s given.)
\end{corollary}

\begin{proof}
    First, note that for any $p$ we can simulate $Bern(p/2)$ from a single $Bern(p)$, by simulating an independent $Bern(1/2)$ and considering the event where both $Bern(p)$ and $Bern(1/2)$ equal $1$.
    Then, it is well-known that $\frac{1}{4x}$ is real analytic on the interval $[1/2, 2/3]$, and $\frac{1}{4x}$ is contained in $[3/8, 1/2]$ for $x \in [1/2, 2/3]$. Since we can generate $Bern(p/2)$, we can therefore apply Theorem \ref{thm:coin_simulator} to get a trial from $Bern\left(\frac{1}{4(p/2)}\right) = Bern\left(\frac{1}{2p}\right)$.
\end{proof}

We now give an algorithm for sampling an edge. The algorithm closely follows the approach from \cite{eden2018sampling} but uses the Bernoulli factory of \cite{nacu2005simulation} to reduce the sampling error in a significantly more efficient way than in \cite{eden2018sampling}. We first give an algorithm for $\Theta(1)$-approximate edge sampling.

Throughout this section, we assume for sake of simplicity that we know the number of edges exactly. The analysis of correctness only uses that we have an upper bound, while the analysis of the complexity needs that we have a lower bound up to a constant factor. Putting this together, it is in fact sufficient to have a constant-factor approximation.

\begin{algorithm}[t]
$u \leftarrow$ uniformly random vertex \label{step:sample_v}\\
$j \leftarrow Unif([\theta]),$ where $\theta = \lceil\sqrt{6m}\rceil$. \label{step:sample_index}\\
Fail if $d(v) > \theta$ or $d(v) \leq j$\\
$v \leftarrow j$-th neighbor of $u$\\
$B \sim Bern(1/3)$ \label{step:bern_third}\\
\uIf{$B = 1$}{
\Return{$uv$}\\
}\ElseIf{$B = 0$ and $v$ is heavy}{
$w \leftarrow$ random neighbor of $v$ \label{step:samp_heavy_nbr}\\
\Return{$vw$}
}
\Return{Fail}
\caption{Sample an edge pointwise $\Theta(1)$-close to uniform} \label{alg:try_to_sample_an_edge}
\end{algorithm}

\begin{lemma} \label{lemma:alg2}
Let $e$ be the edge returned by \Cref{alg:try_to_sample_an_edge} if successful. Then for any light edge $e'$, it holds that $\BP(e = e') = 1/(3 n \theta)$, and for any heavy edge, it holds that $\BP(e = e') = \frac{2}{3} \cdot \frac{d_{\ell}(v)}{d(v)} \cdot \frac{1}{n\theta}$.
\end{lemma}
\begin{proof}
%\jakub{Talya, it would be amazing if you could write this proof. It should be essentially the same as the first half of the proof from your paper.}
Fix a light edge $e'=uv$. Recall that by definition, $uv$ is light iff  $d(u)\leq \theta$ for $\theta=\lceil\sqrt{6m}\rceil$.
The edge $uv$ is returned only in the case that (1) $u$ is sampled in Step~\ref{step:sample_v}, (2) the chosen index $j$ in Step~\ref{step:sample_index} is the label of $v$, and (3) $B=1$ in Step~\ref{step:bern_third}. Therefore,
$\Pr[e=e']=\frac{1}{n}\cdot \frac{1}{\lceil \sqrt{6m}\rceil}\cdot\frac{1}{3}=\frac{1}{3n\theta}$.

Now fix a heavy edge $e'=vw$. The edge $vw$ is returned in the event that (1) the sampled vertex $u$ in Step~\ref{step:sample_v} is a light neighbor of $v$, (2)  the chosen index $j$ in Step~\ref{step:sample_index} is the label of $v$, (3) $B=0$ in Step~\ref{step:bern_third},  and (4) $w$ is the sampled neighbor in Step~\ref{step:samp_heavy_nbr}.
Therefore, if we define $\Gamma_L(v)$ to be the set of light neighbors of $v$, then
$\Pr[e=e']=
\sum_{u\in \Gamma_{L}(v)}
\frac{1}{n}\cdot \frac{1}{\lceil \sqrt{6m}\rceil}\cdot\frac{2}{3}\cdot \frac{1}{d(v)}=\frac{2}{3}\cdot \frac{d_{\ell}(v)}{d(v)}\cdot \frac{1}{n\theta}$.
\end{proof}

We are now able to give an algorithm for sampling an edge perfectly uniformly. Simply re-running the above algorithm until it succeeds would result in $\Theta(1)$-pointwise close to uniform sampling. The algorithm below differs in that if \Cref{alg:try_to_sample_an_edge} returns a heavy edge (which has some bias), we use rejection sampling based on Bernoulli factories to reduce the bias.

\begin{algorithm}[t] 
\Repeat{}{
$vw \leftarrow$ \Cref{alg:try_to_sample_an_edge}\\
\If{$v$ is light}{
\Return{vw}
}
\If{$v$ is heavy}{
Let $w_1, \dots, $ be random neighbors of $v$ \label{step:sample_nbrs}\\
$Y \leftarrow $ use \Cref{cor:simulate} on Bernoulli trials defined as $B_i = [d(w_i) \leq \sqrt{6m}]$\\
\If{$Y = 0$}{
\Return{vw}
}
}
}
\caption{Sample an edge $1\pm \varepsilon$-pointwise-close to uniform} \label{alg:sample_an_edge}
\end{algorithm}

%\begin{theorem}
%\Cref{alg:sample_an_edge} returns an edge at random pointwise $\varepsilon$-close to uniform. Its expected complexity is $O(n/\sqrt{m} + \log \varepsilon^{-1})$.
%\end{theorem}
\begin{theorem}
%There is an algorithm that
\Cref{alg:sample_an_edge} returns a perfectly uniform edge. Its expected complexity is $O(n/\sqrt{m})$.
\end{theorem}

\begin{proof}
We start with proving the correctness of the algorithm. 
By Lemma~\ref{lemma:alg2},  each invocation of~\Cref{alg:try_to_sample_an_edge}
returns each light edge with probability 
 $\frac{1}{3n\theta}$, and each heavy edge with probability $\frac{2}{3} \cdot \frac{d_{\ell}(v)}{d(v)} \cdot \frac{1}{n\theta}$. 
If~\Cref{alg:try_to_sample_an_edge} returns a heavy edge $vw$, then for every $w_i$ sampled in Step~\ref{step:sample_nbrs} in~\Cref{alg:sample_an_edge}, it holds that the indicator of the event $[d(w_i) \leq \lceil \sqrt{6m} \rceil]$ is the result of a Bernoulli trial $Bern(p)$ with $p=\frac{d_{\ell}(v)}{d(v)}$. Let $H$ denote the set of vertices with degree greater than $\lceil \sqrt{6m} \rceil$. 
Then $deg_H(v)\leq |H|\leq \frac{2m}{\lceil \sqrt{6m} \rceil} \leq \sqrt{\frac{2}{3}m}\leq \frac{1}{3}d(v)$, where the last is since $v$ is heavy (so $d(v) > \lceil \sqrt{6m} \rceil)$.
Therefore, $p=\frac{d_{\ell}(v)}{d(v)}\in [\frac{2}{3},1]$.
Hence, by \Cref{cor:simulate}, the value $Y$ returned by \Cref{alg:sample_an_edge} has distribution
$Y \sim \text{Bern}(\frac{1}{2p})$.
Therefore, in a single iteration of the repeat loop, every fixed heavy edge $vw$ is returned with probability 
$\frac{2}{3} \cdot \frac{d_{\ell}(v)}{d(v)} \cdot \frac{1}{n\theta}\cdot \left(\frac{1}{2p}\right)\in \frac{1}{3n\theta}$, as $p = \frac{d_\ell(v)}{d(v)}$.

Therefore, every edge is sampled with probability exactly $\frac{1}{3n},$ so conditioning on some edge being returned, each edge is returned with probability in $\frac{1}{m}$, as claimed. 

We turn to analyze the complexity of the algorithm. 
By the above analysis, every invocation of the loop returns an edge with probability at least $m \cdot \frac{1}{3n \theta} \ge \frac{\sqrt{m}}{10 n}.$ Also, note that each invocation is independent.
Therefore, the expected number of iterations until an edge is returned is $O(n/\sqrt m)$. Furthermore, each invocation of the loop invokes \Cref{alg:try_to_sample_an_edge} once, and Corollary \ref{cor:simulate} at most once. \Cref{alg:try_to_sample_an_edge} clearly takes a constant number of queries.
%Therefore, an iteration where no edge is returned by \Cref{alg:try_to_sample_an_edge} takes a constant number of queries. 
If \Cref{alg:try_to_sample_an_edge} returns a heavy edge, then sampling the $w_i$ neighbors in Step~\ref{step:sample_nbrs} takes $O(1)$ queries in expectation. 
%Observe that if indeed such an edge is returned by \Cref{alg:try_to_sample_an_edge}, then by \Cref{lem:simulate}, it is returned by the algorithm with probability $\frac{1}{2p}\geq \frac{1}{2}$ (since $p\in [2/3,1]$). 
Therefore, the expected number of queries in each loop
is constant. Hence, the expected query complexity is $\Theta\left(n/\sqrt m\right)$.  
\end{proof}

\section*{Acknowledgments}
We thank Nima Anari and Peter Occil for informing us about the existence of Bernoulli factories in the literature, and Nima for pointing us to the reference \cite{nacu2005simulation}.
%
%Talya Eden is supported by the NSF TRIPODS program, award CCF-1740751.
%Shyam Narayanan is supported by the NSF GRFP Fellowship, and the NSF TRIPODS Program (award DMS-2022448).
%Jakub T\v{e}tek is supported by the VILLUM Foundation grant 16582.

\bibliographystyle{plainnat}
\bibliography{literature}

\begin{thebibliography}{13}
\providecommand{\natexlab}[1]{#1}
\providecommand{\url}[1]{\texttt{#1}}
\expandafter\ifx\csname urlstyle\endcsname\relax
  \providecommand{\doi}[1]{doi: #1}\else
  \providecommand{\doi}{doi: \begingroup \urlstyle{rm}\Url}\fi

\bibitem[Addanki et~al.(2022)Addanki, McGregor, and Musco]{addanki2022non}
Raghavendra Addanki, Andrew McGregor, and Cameron Musco.
\newblock Non-adaptive edge counting and sampling via bipartite independent set
  queries.
\newblock \emph{arXiv preprint arXiv:2207.02817}, 2022.

\bibitem[Aliakbarpour et~al.(2018)Aliakbarpour, Biswas, Gouleakis, Peebles,
  Rubinfeld, and Yodpinyanee]{aliakbarpour}
Maryam Aliakbarpour, Amartya~Shankha Biswas, Themis Gouleakis, John Peebles,
  Ronitt Rubinfeld, and Anak Yodpinyanee.
\newblock Sublinear-time algorithms for counting star subgraphs via edge
  sampling.
\newblock \emph{Algorithmica}, 80\penalty0 (2):\penalty0 668--697, 2018.

\bibitem[Assadi et~al.(2019)Assadi, Kapralov, and Khanna]{AKK19}
Sepehr Assadi, Michael Kapralov, and Sanjeev Khanna.
\newblock A simple sublinear-time algorithm for counting arbitrary subgraphs
  via edge sampling.
\newblock In \emph{Innovations in Theoretical Computer Science Conference
  {ITCS}}, volume 124 of \emph{LIPIcs}, pages 6:1--6:20. Schloss Dagstuhl -
  Leibniz-Zentrum fuer Informatik, 2019.

\bibitem[Bhattacharya et~al.(2022)Bhattacharya, Bishnu, Ghosh, and
  Mishra]{bhattacharya2022faster}
Anup Bhattacharya, Arijit Bishnu, Arijit Ghosh, and Gopinath Mishra.
\newblock Faster counting and sampling algorithms using colorful decision
  oracle.
\newblock In \emph{39th International Symposium on Theoretical Aspects of
  Computer Science (STACS 2022)}. Schloss Dagstuhl-Leibniz-Zentrum f{\"u}r
  Informatik, 2022.

\bibitem[Biswas et~al.(2021)Biswas, Eden, and Rubinfeld]{BER21}
Amartya~Shankha Biswas, Talya Eden, and Ronitt Rubinfeld.
\newblock Towards a decomposition-optimal algorithm for counting and sampling
  arbitrary motifs in sublinear time.
\newblock In \emph{Approximation, Randomization, and Combinatorial
  Optimization. Algorithms and Techniques, {APPROX/RANDOM} 2021, to appear},
  2021.

\bibitem[Eden and Rosenbaum(2018)]{eden2018sampling}
Talya Eden and Will Rosenbaum.
\newblock {On Sampling Edges Almost Uniformly}.
\newblock In Raimund Seidel, editor, \emph{1st Symposium on Simplicity in
  Algorithms (SOSA 2018)}, volume~61 of \emph{OpenAccess Series in Informatics
  (OASIcs)}, pages 7:1--7:9, Dagstuhl, Germany, 2018. Schloss
  Dagstuhl--Leibniz-Zentrum fuer Informatik.
\newblock ISBN 978-3-95977-064-4.
\newblock \doi{10.4230/OASIcs.SOSA.2018.7}.
\newblock URL \url{http://drops.dagstuhl.de/opus/volltexte/2018/8300}.

\bibitem[Eden et~al.(2019)Eden, Ron, and Rosenbaum]{eden2019arboricity}
Talya Eden, Dana Ron, and Will Rosenbaum.
\newblock {The Arboricity Captures the Complexity of Sampling Edges}.
\newblock In Christel Baier, Ioannis Chatzigiannakis, Paola Flocchini, and
  Stefano Leonardi, editors, \emph{46th International Colloquium on Automata,
  Languages, and Programming (ICALP 2019)}, volume 132 of \emph{Leibniz
  International Proceedings in Informatics (LIPIcs)}, pages 52:1--52:14,
  Dagstuhl, Germany, 2019. Schloss Dagstuhl--Leibniz-Zentrum fuer Informatik.
\newblock ISBN 978-3-95977-109-2.
\newblock \doi{10.4230/LIPIcs.ICALP.2019.52}.
\newblock URL \url{http://drops.dagstuhl.de/opus/volltexte/2019/10628}.

\bibitem[Eden et~al.(2021)Eden, Mossel, and Rubinfeld]{eden_multiple}
Talya Eden, Saleet Mossel, and Ronitt Rubinfeld.
\newblock Sampling multiple edges efficiently.
\newblock In \emph{Approximation, Randomization, and Combinatorial
  Optimization. Algorithms and Techniques (APPROX/RANDOM 2021)}. Schloss
  Dagstuhl-Leibniz-Zentrum f{\"u}r Informatik, 2021.

\bibitem[Feige(2006)]{feige2006sums}
Uriel Feige.
\newblock On sums of independent random variables with unbounded variance and
  estimating the average degree in a graph.
\newblock \emph{SIAM Journal on Computing}, 35\penalty0 (4):\penalty0 964--984,
  2006.

\bibitem[Fichtenberger et~al.(2020)Fichtenberger, Gao, and Peng]{Peng20}
Hendrik Fichtenberger, Mingze Gao, and Pan Peng.
\newblock Sampling arbitrary subgraphs exactly uniformly in sublinear time.
\newblock In Artur Czumaj, Anuj Dawar, and Emanuela Merelli, editors,
  \emph{47th International Colloquium on Automata, Languages, and Programming,
  {ICALP} 2020, July 8-11, 2020, Saarbr{\"{u}}cken, Germany (Virtual
  Conference)}, volume 168 of \emph{LIPIcs}, pages 45:1--45:13. Schloss
  Dagstuhl - Leibniz-Zentrum f{\"{u}}r Informatik, 2020.
\newblock \doi{10.4230/LIPIcs.ICALP.2020.45}.
\newblock URL \url{https://doi.org/10.4230/LIPIcs.ICALP.2020.45}.

\bibitem[Lapinskas et~al.(2019)Lapinskas, Dell, and
  Meeks]{lapinskas2019approximately}
John~A Lapinskas, Holger Dell, and Kitty Meeks.
\newblock Approximately counting and sampling small witnesses using a colourful
  decision oracle.
\newblock In \emph{ACM-SIAM Symposium on Discrete Algorithms (SODA20)}, 2019.

\bibitem[Nacu and Peres(2006)]{nacu2005simulation}
\c{S}erban Nacu and Yuval Peres.
\newblock Fast simulation of new coins from old.
\newblock \emph{SIAM Journal on Computing}, 35\penalty0 (4):\penalty0 964--984,
  2006.

\bibitem[T\v{e}tek and Thorup(2022)]{tetek2022edge}
Jakub T\v{e}tek and Mikkel Thorup.
\newblock Edge sampling and graph parameter estimation via vertex neighborhood
  accesses.
\newblock In \emph{Proceedings of the 54th Annual ACM SIGACT Symposium on
  Theory of Computing}, STOC 2022, page 1116–1129, New York, NY, USA, 2022.
  Association for Computing Machinery.
\newblock ISBN 9781450392648.
\newblock \doi{10.1145/3519935.3520059}.
\newblock URL \url{https://doi.org/10.1145/3519935.3520059}.

\end{thebibliography}
\end{document}